\documentclass[letterpaper, 10 pt, conference]{ieeeconf}  

\IEEEoverridecommandlockouts                              
\usepackage{cite}
\usepackage{url}

\usepackage{graphicx}
\usepackage{algorithm}
\usepackage[noend]{algorithmic}
\usepackage{subfig}
\usepackage{amssymb, amsmath,graphicx,charter, latexsym}
\usepackage{enumerate}

\newtheorem{definition}{Definition}
\newtheorem{lemma}{Lemma}

\newtheorem{theorem}{Theorem}

\title{\LARGE \bf
Real-Time Stochastic Processing Networks with Concurrent Resource Requirements
}

\author{I-Hong Hou and Rahul Singh
\thanks{This material is based upon work partially supported by NSF under Contracts CNS-1035378, CNS-1035340, and CCF-0939370, and USARO under Contract Nos. W911NF-08-1-0238 and W-911-NF-0710287.}
\thanks{I-Hong Hou is with Computer Engineering and Systems Group, Department of ECE, Texas A\&M University, College Station, TX 77843, USA.
        {\tt\small ihou@tamu.edu}}%
\thanks{Rahul Singh is with Computer Engineering and Systems Group, Department of ECE, Texas A\&M University, College Station, TX 77843, USA.
        {\tt\small rsingh1@tamu.edu}}%
}
\begin{document}
\maketitle
\thispagestyle{empty}
\pagestyle{empty}

\begin{abstract}
Stochastic Processing Networks (SPNs) can be used to model communication networks, manufacturing systems, service systems, etc. We consider a real-time SPN where tasks generate jobs with strict deadlines according to their traffic patterns. Each job requires the concurrent usage of some resources to be processed. The processing time of a job may be stochastic, and may not be known until the job completes. Finally, each task may require that some portion of its tasks to be completed on time.

In this paper, we study the problem of verifying whether it is feasible to fulfill the requirements of tasks, and of designing scheduling policies that actually fulfill the requirements. We first address these problems for systems where there is only one resource. Such systems are analog to ones studied in a previous work, and, similar to the previous work, we can develop sharp conditions for feasibility and scheduling policy that is feasibility-optimal. We then study systems with two resources where there are jobs that require both resources to be processed. We show that there is a reduction method that turns systems with two resources into equivalent single-resource systems. Based on this method, we can also derive sharp feasibility conditions and feasibility-optimal scheduling policies for systems with two resources.
\end{abstract}

\section{Introduction}  \label{section:introduction}

Stochastic Processing Networks (SPNs), which were proposed by Harrison \cite{mh02, mh03}, consider systems where jobs compete for resources to be processes. SPNs provide a more general model than queueing networks in that they allow a job to require the concurrent usage of multiple resources to be processed. With the more general and flexible model, SPNs can be used to describe a large range of systems, including communication systems, manufacturing systems, service systems, etc.

Most existing work on SPNs, such as Jiang and Walrand \cite{lj09} and Dai and Lin\cite{jgd05}, focuses on stabilizing the system, that is, making the number of unfinished jobs in the system bounded. On the other hand, little is known about the delay of each job. As more and more applications of SPNs require hard delay bound for each job, it becomes increasingly important to address the issue of per-job delay.

There has been much work on providing per-job delay bounds in real-time systems. Liu and Layland \cite{cll73} has considered the scheduling problem for providing per-job delay guarantees in  a single-processor environment, and has proposed the well-known \emph{earliest deadline first} scheduling policy. The system considered in \cite{cll73} can be thought of as a special case of SPNs where there is only one resource in the system and the processing time of each job is given. Recently, there is an emerging theory \cite{IH09, IHH09MobiHoc} for real-time wireless communications, where clients compete for the shared wireless channel, and the unreliable nature of wireless transmissions is considered. Such a system corresponds to SPNs where there is only one resource, time is slotted, and the processing time of each job is stochastic and unknown.

In this paper, we study real-time SPNs for continuous-time systems. We consider the fact that a job may require multiple resources to be processed, as in the general model of SPNs. We also consider the delay bound of each job, and that the processing time of each job is stochastic and unknown, as in \cite{IH09, IHH09MobiHoc}. Further, we assume that each task requires at least some portion of its jobs to be completed on time. Based on this model, we study the problem of verifying whether it is feasible to fulfill the requirements of all tasks in the system, and, if the system is feasible, designing scheduling policies that actually fulfill all the requirements.

We first study the case where there is only one resource in the system. We, similar to \cite{IH09, IHH09MobiHoc}, derive sharp conditions for feasibility and a scheduling policy that is feasibility-optimal in that it fulfills all feasible systems.

We then study a special case where there are two resources in the system, and some jobs may require both resources to be processed. We show that there is a reduction method that transforms this system to an equivalent single-resource system. Therefore, we can also obtain sharp feasibility conditions and a feasibility-optimal scheduling policy for the two-resource system.

The rest of the paper is organized as follows: Section \ref{section:model} introduces our model for real-time SPNs. Section \ref{section:single} discusses the feasibility conditions and feasibility-optimal scheduling policy for systems with only one resource. Section \ref{section:two} proposes a reduction method for systems with two resources, and demonstrates the usage of this method by deriving feasibility conditions and scheduling policy for such systems. Finally, Section \ref{section:conclusion} concludes the paper.

\section{System Model}  \label{section:model}

In this section, we introduce our system model. The model extends one that is proposed in previous works \cite{IH09, IHH09MobiHoc}, which only consider discrete-time systems with one resource.

Consider a continuous-time system that consists of several \emph{tasks} and \emph{resources}. We denote the set of tasks by $N$, and the set of resources by $L$. Each task generates \emph{jobs} that need to be processed. Jobs of each task $n\in N$ requires the concurrent usage of some of the resources to be processed, and we denote by $L(n)$ the set of resources that jobs of task $n$ require. Each resource can only be used by one job at any point of time. Therefore, a job of task $n_1$ and a job of task $n_2$ can be processed at the same time only when $L(n_1)\cap L(n_2)=\phi$. We assume that the system is preemptive, and there is a centric server that selects a set of jobs to be processed at each point of time. Figure \ref{figure:system1} depicts one such system. In the depicted system, jobs of task 1 and of task 4 can be processed simultaneously. Jobs of task 2 and of task 4 can also be processed simultaneously. However, jobs of task 1 and of task 3 cannot be processed simultaneously, as they both require resource 1.

\begin{figure}	
\includegraphics[width=3.3in]{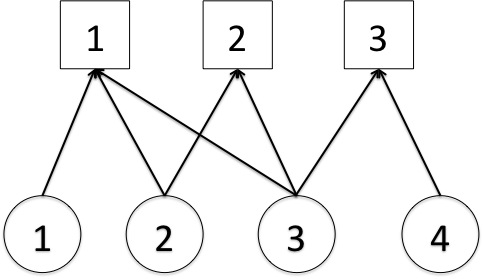}
\caption{A system with 3 resources and 4 tasks. We use circles to represent tasks and squares to represent resources. A line connecting a resource and a task indicates that the task requires the resource to be processed. In the system, we have $L(1)=\{1\},L(2)=\{1,2\}, L(3)=\{1,2,3\}$, and $L(4)=\{3\}$.}\label{figure:system1}
\end{figure}

We assume that time is divided into \emph{frames}, where the length of each frame is $T$ time units. At the beginning of each frame, each task may, or may not, generates a job. The generation of jobs of each task depends on its traffic pattern, which we assume can be modeled as an irreducible finite state Markov chain. The long-term average number of job generations per frame for each task can then be computed, and we denote it by $r_n$. We can also compute the long-term average probability that a subset $S$ of tasks generate jobs, and none of the other tasks do, in a frame, which we denote by $r(S)$.

We assume that the processing times of jobs of task $n$ are independent exponential random variables with mean $\frac{1}{\lambda_n}$. Moreover, we assume that, although we know the distribution of the processing time of a job, the exact value of the processing time cannot be known before the completion of the job. Due to the memoryless property of exponential random variable, we can conclude that, when a job of task $n$ is being processed at some point of time, the probability that it is completed in the next $\Delta t$ time units is $1-e^{-\lambda_n\Delta t}$, regardless how much time the job has been processed before.

We consider the hard delay bound of real-time SPNs. In particular, we say that each job needs to be completed within $T$ time units after it is generated. In other words, jobs that are generated at the beginning of some frame need to be completed before the end of the frame. As the processing times of jobs are stochastic, it may be the case that some jobs cannot be completed before their deadlines, in which case we say that these jobs \emph{expire}, and remove them from the system.

We measure the performance of a task by its \emph{timely-throughput}, which measures the long-term average number of completed jobs per frame for each task:
\begin{definition}
Let $e_n(k)$ be the indicator function that a job of $n$ is completed in frame $k$. The \emph{timely-throughput} of task $n$ is defined to be $\liminf_{K\rightarrow\infty}\frac{\sum_{k=1}^Ke_n(k)}{K}$.
\end{definition}

We assume that each task imposes a \emph{timely-throughput requirement}, denoted by $q_n$, and requires that $\liminf_{K\rightarrow\infty}\frac{\sum_{k=1}^Ke_n(k)}{K}\geq q_n$. Since, on average, task $n$ generates $r_n$ jobs per frame, the timely-throughput requirement is equivalent to one that requires that at least $q_n/r_n$ of task $n$'s jobs to be completed on time.

In this paper, we investigate the problem of evaluating whether it is \emph{feasible} to \emph{fulfill} a system, and of designing \emph{feasibility-optimal} policies for scheduling jobs.

\begin{definition}
A system is said to be \emph{fulfilled} by a scheduling policy if, under this policy, $\liminf_{K\rightarrow\infty}\frac{\sum_{k=1}^Ke_n(k)}{K}\geq q_n$, for all $n$.
\end{definition}

\begin{definition}
A system is \emph{feasible} if there exists some scheduling policy that fulfils it.
\end{definition}

\begin{definition}
A scheduling policy is \emph{feasibility-optimal} if it fulfils all feasible systems.
\end{definition}

\section{Feasibility Conditions and Scheduling Policy for Single-Resource Systems} \label{section:single}

In this section, we discuss a special case where there is only one resource in the system, that is, $L = \{1\}$ and $L(n) = 1$, for all $n$. Hence, at any point of time, at most one job can be processed. Figure \ref{figure:system2} shows an example of single-resource system.

\begin{figure}	
\includegraphics[width=3.3in]{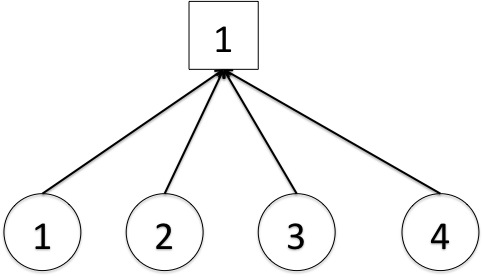}
\caption{A single-resource system.}\label{figure:system2}
\end{figure}

We derive feasibility conditions and a feasibility-optimal scheduling policy for such systems. The results presented in this section can be derived by using similar techniques as those in previous works that consider discrete-time systems. Hence, we omit all the proofs in this section.

We first observe that the timely-throughput requirement of each task poses a constraint on the long-term average amount of time per frame that jobs of the task should be processed. 

\begin{lemma}\label{l1}
The long-term average timely-throughput of task $n$ is at least $q_{n}$ jobs per frame if and only if the long-term average amount of time that jobs of task $n$ are processed is at least $w_{n} = \frac{q_{n}}{\lambda_{n}}$ per frame.
\end{lemma}
\begin{proof}
This is analog to Lemma 1 in \cite{IHH09MobiHoc}.
\end{proof}

We call $w_n$ the \emph{implied work load} of task $n$. Since the job arrivals and processing time are random, there are times that system is forced to be \emph{idle} and does not process any jobs. This can happen either because none of the tasks generate any job in a frame, or because all jobs generated in a frame are completed before the end of the frame, and hence there are no jobs to be processed for the rest of the frame. It can be shown that the long-term average amount of time per frame that the system is idle is invariant for all \emph{work-conserving} policies.

\begin{definition}
A scheduling policy is said to be \emph{work-conserving} if it never idles whenever there is a job in the system that is not completed yet.
\end{definition}

Obviously, a policy cannot lose optimality by processing more jobs, and hence there is a feasibility-optimal policy that is work-conserving. From now on, we limit our discussions on work-conserving policies. As the amount of idle time is the same for all work-conserving policies, we can define the average amount of idle time per frame when only a subset $S$ of tasks is present in the system to be $E[I_{S}]$. Specifically, let $g_n(k)$ be the indicator function that task $n$ generates a job in frame $k$, and $t_n(k)$ be the processing time of the job of task $n$ in frame $k$, which is an exponential random variable with mean $\frac{1}{\lambda_n}$, we have \[E[I_S]:=\lim_{K\rightarrow\infty}\frac{\sum_{k=1}^KE[(T-\sum_{n\in S}g_n(k)t_n(k))^+]}{K},\]
where $x^+:=\max\{x,0\}$.

The maximum possible amount of time that the system spends processing jobs of tasks in $S$ can now be expressed as $T-E[I_{S}]$, which can be achieved by always processing jobs of tasks in $S$ prior to other jobs. Further, the total amount of implied work load of tasks in $S$ is $\sum_{n\in S}w_n$. Hence, for a system to be feasible, we require that $\sum_{n\in S}w_n\leq T-E[I_S]$, for all $S\subseteq N$. It turns out that this condition is both necessary and sufficient.

\begin{theorem}\label{l2}
A system is feasible if and only if \[\sum_{n\in S} w_{n} + E[I_{S}]\leq T,\] holds for every subset $S\subseteq N$.
\end{theorem}
\begin{proof}
This is analog to Theorem 4 in \cite{IHH09MobiHoc}.
\end{proof}

Next, we propose our scheduling policy. The policy is based on the concept of \emph{time-based debt}.
\begin{definition}
Let $\gamma_n(k)$ be the amount of time that the job of task $n$ is processed in frame $k$. The \emph{time-based debt} of task $n$ at frame $k$ is defined as
 \[d_n(k):=(k-1)w_n-\sum_{j=1}^{k-1}\gamma_n(j).\]
\end{definition}

The time-based debt can be interpreted as the amount of time which the task $n$ is lagging behind that required by its implied work load. We can establish a sufficient condition for a policy to be feasibility-optimal based on the time-based debt.

\begin{theorem} \label{theorem:max}
A scheduling policy which maximizes $E\{\gamma_n(k)d_n(k)^+\}$ for all $k$ is feasibility optimal.
\end{theorem}
\begin{proof}
This is analog to Theorem 3 in \cite{IH10}
\end{proof}

It turns out that there exists a simple online policy that satisfies the above condition. We call the policy the \emph{Largest Debt First} policy.
\begin{definition}
The \emph{Largest Debt First} policy computes the time-based debt for each task at the beginning of each frame and decides priorities of the tasks based on the debts. The policy gives a higher priority to a task with a higher time-based debt. The job of task $n$ is then scheduled to be processed only when the jobs of all the tasks having higher priorities than task $n$ are processed.
\end{definition}
\begin{theorem} \label{theorem:opt}
The Largest Debt First policy maximizes $E\{\gamma_n(k)d_n(k)^+\}$, and hence is feasibility-optimal.
\end{theorem}

\section{A Reduction Method for Systems with Two Resources} \label{section:two}

In this section, we discuss a special case where there are two resources in the system. We show that there exists a reduction method that transforms the system into an equivalent single-resource system. Therefore, we can obtain results for feasibility conditions and scheduling policies.

We consider a system with two resources and several tasks. Jobs of task 1 only requires resource 1 to be processed, jobs of task 2 only requires resource 2 to be processed, and all other jobs require the concurrent usage of both resource 1 and resource 2 to be processed. Therefore, at any point of time, we can either schedule a job of task 1 and a job of task 2 concurrently, or schedule a single job to be processed. We assume that each task generates one job in each frame. Figure \ref{figure:system3} shows an example of such systems.

\begin{figure}	
\includegraphics[width=3.3in]{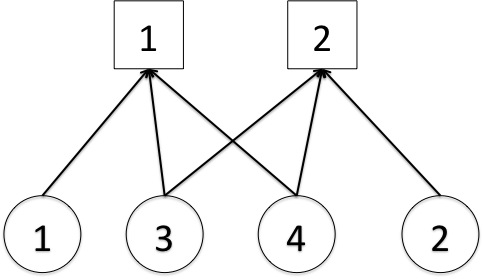}
\caption{A system with two resources.}\label{figure:system3}
\end{figure}

We now introduce our reduction method. We create a single-resource system. In this system, each task corresponds to a set of tasks that can be processed simultaneously in the original two-resource system. We denote the task that corresponds to the set containing $\{n\}$ by task $n^*$, and denote the task that corresponds to the set $\{1,2\}$, which is the only set that contains more than one tasks, by $c^*$. The set of tasks in the one-resource system is denoted by $N^*$.

We aim to construct the single-resource system so that whenever a set of jobs in the two-resource system is selected for process, the corresponding job in the single-resource system is also selected. We observe that, when a job of task 1 in the two-resource system is selected for processing, a feasibility-optimal policy should also select the job of task 2 for processing as long as the job of task 2 is not completed yet. Therefore, if the policy selects only the job of task 1 for processing at some point in a frame, we have that the processing time of the job of task 1 is larger than that of the job of task 2 in this frame, which occurs with probability $\frac{\lambda_1}{\lambda_1+\lambda_2}$. We hence set the probability that task $1^*$ generates a job in a frame to be $\frac{\lambda_1}{\lambda_1+\lambda_2}$. Similarly, we set the probability that task $2^*$ generates a job in a frame to be $\frac{\lambda_2}{\lambda_1+\lambda_2}$. Moreover, in a frame, only one of task $1^*$ and task $2^*$ generates a job. Tasks other than $1^*$ and $2^*$ generate jobs in every frame.

Next, we discuss the processing times of jobs in the single-resource system. In a frame, let $t_n$ be the processing time of the job of task $n$ in the two-resource system. We have that $t_n$ is an exponential random variable with mean $1/\lambda_n$. Moreover, given $t_1>t_2$, $(t_1-t_2)$ is an exponential random variable with mean $1/\lambda_1$. Therefore, we set the processing times of jobs of task $1^*$ in the single-resource system to be exponential random variables with mean $1/\lambda_1$. Similarly, the processing times of jobs of task $2^*$ are exponential random variables with mean $1/\lambda_2$. Further, we have that $\min\{t_1,t_2\}$, which is the amount of time that both jobs of task 1 and task 2 can be processed simultaneously, is an exponential random variable with mean $\frac{1}{\lambda_1+\lambda_2}$. Hence, we set the processing times of jobs of task $c^*$ to be exponential random variables  with mean $\frac{1}{\lambda_1+\lambda_2}$. Finally, the processing times of jobs of other tasks $n^*$ are exponential random variables with mean $1/\lambda_n$.

We then address the timely-throughput requirements in the single-resource system. Let $q_{n^*}$ be the timely-throughput requirement of task $n^*$. The event of the completion of a job of task $c^*$ corresponds to the event that one of the jobs of task 1 and task 2 is completed in the two-resource system. When the event occurs, the probability that the job completion is of task 1 is $\frac{\lambda_2}{\lambda_1+\lambda_2}$, and the probability that the job completion is of task 2 is $\frac{\lambda_2}{\lambda_1+\lambda_2}$. The event of the completion of a job of task $n^*$ other than $c^*$ corresponds to a job completion of task $n$ in the two-resource system. Therefore, we choose $q_{1^*}, q_{2^*},\dots, q_{c^*}$ such that
\begin{align}
\frac{\lambda_2}{\lambda_1+\lambda_2}q_{c^*}+q_{1^*}&\geq q_1, \label{eq:reduction:1}\\
\frac{\lambda_1}{\lambda_1+\lambda_2}q_{c^*}+q_{2^*}&\geq q_2,  \label{eq:reduction:2}\\
q_{n^*}&\geq q_n, \forall n\notin\{1,2,c\}.   \label{eq:reduction:3}
\end{align}
Finally, for the single-resource system to be feasible, we require that
\begin{align}
\sum_{n^*\in S^*}\frac{q_{n^*}}{\lambda_{n^*}} + E[I_{S^*}]\leq T, \forall S^*\subseteq \{1^*, 2^*, \dots, c^*\},  \label{eq:reduction:4}\\
q_{n^*}\geq 0, \forall n^*\in\{1^*, 2^*, \dots, c^*\},\label{eq:reduction:5}
\end{align}
where we set $\lambda_{c^*}=\lambda_1+\lambda_2$, and $\lambda_{n^*}=\lambda_n$, for all $n^*$ other than $c^*$.

We show that the two-resource system is feasible if there exists a corresponding feasible single-resource system.

\begin{theorem}
The two-resource system is feasible if and only if there exists a vector $[q_{1^*}, q_{2^*}, \dots, q_{c^*}]$ that satisfy (\ref{eq:reduction:1})--(\ref{eq:reduction:5}).
\end{theorem}
\begin{proof}
We first show that the existence of the desired vector $[q_{1^*}, q_{2^*}, \dots, q_{c^*}]$ is necessary for the two-resource system to be feasible. Suppose the two-resource system is feasible and is fulfilled by a policy $\eta$. We can assume that when $\eta$ schedules the job of task 1, it also schedules the job of task 2 as long as it has not been completed yet, and vice versa. As explained above, we can construct the one-resource system so that whenever $\eta$ schedules a subset of jobs, the corresponding job in the single-resource system is scheduled. We choose $q_{n^*}$ to be the resulting timely-throughput of task $n^*$ in the single-resource system. Then, the vector $[q_{1^*},q_{2^*},\dots, q_{c^*}]$ is feasible and satisfy (\ref{eq:reduction:4})--(\ref{eq:reduction:5}), as they are achieved by $\eta$. Moreover, we have that, in the two-resource system, the timely-throughput of task 1 is $\frac{\lambda_2}{\lambda_1+\lambda_2}q_{c^*}+q_{1^*}$, that of task 2 is $\frac{\lambda_1}{\lambda_1+\lambda_2}q_{c^*}+q_{2^*}$, and that of task $n^*$ is $q_{n^*}$, for all other $n^*$. As $\eta$ fulfills the two-resource system, (\ref{eq:reduction:1})--(\ref{eq:reduction:3}) are also satisfied.

Next, we show that the existence of a desired vector $[q_{1^*}, q_{2^*}, \dots, q_{c^*}]$ is also sufficient for the two-resource system to be feasible. Suppose there exists some vectors that satisfy (\ref{eq:reduction:1})--(\ref{eq:reduction:5}), we choose $[q_{1^*}, q_{2^*}, \dots, q_{c^*}]$ to be the one with the largest $q_{c^*}$ among those that satisfy (\ref{eq:reduction:1})--(\ref{eq:reduction:5}). Since $[q_{1^*}, q_{2^*}, \dots, q_{c^*}]$ is feasible for the single-resource system, there exists a policy $\eta^*$ that fulfills the system. Similar to the previous paragraph, we only need to show that the scheduling decisions of $\eta^*$ correspond to ones for the two-resource system.

Recall that scheduling the job of task $1^*$, or of task $2^*$, corresponds to the event that the job of task $1$, or of task $2$ is scheduled after the job of task $2$, or of task $1$ is completed, respectively. Therefore, a schedule for the single-resource system does not correspond to any schedule for the two-resource system if the job of task $1^*$ or task $2^*$ is scheduled before the job of task $c^*$ is completed. It is easy to check that all other schedules for the single-resource system correspond to ones for the two-resource system. Hence, it suffices to show that, under $\eta^*$, the long-term average amount of time that a job of task $1^*$ or $2^*$  is scheduled before the completion of a job of task $c^*$ is zero.

Consider a vector $[q'_{1^*},q'_{2^*},\dots, q'_{c^*}]$ with $q'_{c^*}=q_{c^*}+\epsilon$, $q'_{1^*}=(q_{1^*}-\frac{\lambda_2}{\lambda_1+\lambda_2}\epsilon)^+$, $q'_{2^*}=(q_{2^*}-\frac{\lambda_1}{\lambda_1+\lambda_2}\epsilon)^+$, and $q'_{n^*}=q_{n^*}$, for all other $n^*$, for some $\epsilon>0$, where $x^+:=\max\{x,0\}$. Since we choose the vector $[q_{1^*},q_{2^*},\dots, q_{c^*}]$ to be one that has the largest $q_{c^*}$ among those that satisfy (\ref{eq:reduction:1})--(\ref{eq:reduction:5}), $[q'_{1^*},q'_{2^*},\dots, q'_{c^*}]$ does not satisfy (\ref{eq:reduction:1})--(\ref{eq:reduction:5}). However, we have
\begin{align}
\frac{\lambda_2}{\lambda_1+\lambda_2}q'_{c^*}+q'_{1^*}&\geq \frac{\lambda_2}{\lambda_1+\lambda_2}q_{c^*}+q_{1^*}\geq q_1,\\
\frac{\lambda_1}{\lambda_1+\lambda_2}q'_{c^*}+q'_{2^*}&\geq\frac{\lambda_1}{\lambda_1+\lambda_2}q_{c^*}+q_{2^*}\geq q_2,\\
q'_{n^*}&=q_{n^*}\geq q_n, \forall n^*\notin\{1^*,2^*,c^*\},\\
q'_{n^*}&\geq0, \forall n^*.
\end{align}
That is, $[q'_{1^*},q'_{2^*},\dots, q'_{c^*}]$ satisfy (\ref{eq:reduction:1})--(\ref{eq:reduction:3}) and (\ref{eq:reduction:5}). Therefore, there exists some subset $\hat{S^*}$ such that $\sum_{n^*\in\hat{S^*}}\frac{q'_{n^*}}{\lambda_{n^*}}+E[I_{\hat{S^*}}]>T$. As $q'_{n^*}\leq q_{n^*}$, for all $n^*$ other than $c^*$, we have $c^*\in \hat{S^*}$. Further, it is easy to check that, for every subset $S^*$ that either contains $1^*$, $2^*$, or both, we have either $\sum_{n^*\in S^*}\frac{q'_{n^*}}{\lambda_{n^*}}+E[I_{S^*}]\leq \sum_{n^*\in S^*}\frac{q_{n^*}}{\lambda_{n^*}}+E[I_{S^*}]\leq T$, or $q'_{n^*}=0$ for $n^*\in S^*\cap\{1^*,2^*\}$. As $E[I_{S^*}]$, the expected amount of forced idle time when the system only works on a subset $S^*$ of tasks, is non-increasing with $S^*$, we can conclude that there exists a subset $\hat{S^*}$ such that $c^*\in \hat{S^*}$, $\{1^*, 2^*\}\cap \hat{S^*}=\phi$, and $\sum_{n^*\in\hat{S^*}}\frac{q'_{n^*}}{\lambda_{n^*}}+E[I_{\hat{S^*}}]>T$. Moreover, such a subset $\hat{S^*}$ exists for every positive $\epsilon$. Hence, there exists a subset $\hat{S^*}$ such that $c^*\in \hat{S^*}$, $\{1^*, 2^*\}\cap \hat{S^*}=\phi$, and $\sum_{n^*\in\hat{S^*}}\frac{q_{n^*}}{\lambda_{n^*}}+E[I_{\hat{S^*}}]=T$.

Recall that $\frac{q_{n^*}}{\lambda_{n^*}}$ is the amount of time needed by task $n^*$ in order to obtain a timely-throughput of $q_{n^*}$, and $T-E[I_{\hat{S^*}}]$ can be interpreted as the amount of time that the system work on tasks in $\hat{S^*}$ if the system never schedules jobs of tasks not in $\hat{S^*}$ before all jobs of tasks in $\hat{S^*}$ are completed.

Now, suppose that, under $\eta^*$, the long-term average amount of time that a job of task $1^*$ or $2^*$  is scheduled before the completion of a job of task $c^*$ is larger than zero. As $\{1^*, 2^*\}\cap \hat{S^*}=\phi$, the amount of time that $\eta^*$ spends on tasks in $\hat{S^*}$ is less than $T-E[I_{\hat{S^*}}]=\sum_{n^*\in\hat{S^*}}\frac{q_{n^*}}{\lambda_{n^*}}$, and it is impossible to fulfill all timely-throughput requirements for tasks in $\hat{S^*}$, which leads to a contradiction. Therefore, we can conclude that all scheduling decisions by $\eta^*$ in the single-resource system correspond to ones in the two-resource system, and the proof is completed.
\end{proof}

Next, we propose a scheduling policy for the two-resource system. The policy is called the \emph{Largest Total Debt First} policy and is very similar to the Largest Debt First policy for the single-resource system. We define the time-based debt the same as the single-resource system. At the beginning of each frame, the policy selects the set of jobs so that the sum of time-based debts of these jobs is maximized and processes them until at least one of the jobs is completed, at which point of time the policy selects another set of jobs that maximize the sum of time-based debts, and so on. We show that the Largest Total Debt First policy is feasibility-optimal for the two-resource system.

\begin{theorem}
The Largest Total Debt First policy is feasibility-optimal for the two-resource system.
\end{theorem}
\begin{proof}
Let $d_n(k)$ be the time-based debt of task $n$ in the $k^{th}$ frame, and $\gamma_n(k)$ to be the amount of time that the system processes the job of task $n$ in the $k^{th}$ frame. Theorem \ref{theorem:max} has shown that a policy that maximizes $E\{\gamma_n(k)d_n(k)^+\}$ is feasibility optimal. Moreover, it is obvious that when a policy selects the job of task 1 for processing, the value of $E\{\gamma_n(k)d_n(k)^+\}$ does not decrease if it also processes the job of task 2 whenever it is available, and vice versa. As discussed as above, a policy that processes the job of task 1, or 2, whenever it is processing the job of task 2, or 1, corresponds to a policy for the single-resource system. Hence, we only need to show that the Largest Total Debt First policy maximizes $E\{\gamma_n(k)d_n(k)^+\}$ among those that correspond to policies for the single-resource system.

We define, for the corresponding single-resource system, $d_{c^*}(k)=d_1(k)^++d_2(k)^+$, and $d_{n^*}(k)=d_n(k)^+$, for all $n^*\neq c^*$. Let $\gamma_{n^*}(k)$ be the amount of time that the job of task $n^*$ is processed under some policy $\eta^*$. We then have that, under the corresponding policy $\eta$ for the two-resource system, $\gamma_1(k)=\gamma_{c^*}(k)+\gamma_{1^*}(k)$, $\gamma_2(k)=\gamma_{c^*}(k)+\gamma_{2^*}(k)$, and $\gamma_n(k)=\gamma_{n^*}(k)$, for all $n\notin\{1,2\}$. Hence, $E\{\gamma_{n^*}(k)d_{n^*}(k)^+\}=E\{\gamma_n(k)d_n(k)^+\}$. Finally, as shown in Theorem \ref{theorem:opt}, the Largest Debt First policy for the single-resource system maximizes $E\{\gamma_{n^*}(k)d_{n^*}(k)^+\}$, and therefore, its corresponding policy in the two-resource system, which is the Largest Total Debt First policy, maximizes $E\{\gamma_n(k)d_n(k)^+\}$.
\end{proof}

\section{Conclusions}   \label{section:conclusion}

We have studied real-time Stochastic Processing Networks in this paper. We have proposed a model for real-time SPNs that jointly consider the concurrent resource usage, the hard delay bound, and the stochastic processing time of jobs, as well as the traffic patterns and timely-throughput requirements of tasks. We have addressed the problem of characterizing feasibility and scheduling jobs for single-resource systems. We have also proposed a reduction method that transforms two-resource systems into equivalent single-resource ones. Based on this method, we have proved a sharp condition for two-resource systems to be feasible. We have also proposed a simple online scheduling policy for two-resource systems that is feasibility-optimal.

\small
\bibliographystyle{ieeetr}
\bibliography{reference}
\end{document}